\documentclass{article}
\usepackage[T1]{fontenc}
\usepackage[pdftex]{color,graphicx}
\usepackage[numbers,square]{natbib}
\usepackage{amsmath,amssymb,amsthm,eucal,upref}
\usepackage{mathtools}
\usepackage{geometry}

\date{}

\theoremstyle{plain}
\newtheorem{theorem}{Theorem}[section]
\newtheorem{lemma}[theorem]{Lemma}

\theoremstyle{definition}

\theoremstyle{definition}
\newtheorem{remark}[theorem]{Remark}
\numberwithin{equation}{section}

\newcommand{\RR}{\mathrm{I\kern-0.20emR}}
\newcommand{\E}{\mathrm{e}\kern0.2pt}%

\def\div{\textrm{\rm div }}
\def\curl{\textrm{\rm curl }}

\newcommand{\ii}{\kern0.05em\mathrm{i}\kern0.05em}%
\newcommand{\D}{\mathrm{d}\kern0.2pt}%
\DeclareMathOperator{\BS}{BS}

\newcommand{\be}{\begin{equation}}
\newcommand{\ee}{\end{equation}}

\newcommand{\dom}{\Omega}
\newcommand{\surf}{\partial\Omega^t}
\newcommand{\bottom}{\partial\Omega^b}
\newcommand{\boundary}{\partial\Omega}

\newcommand{\op}[1]{\skew{4}\hat{#1}}

\renewcommand{\v}[1]{\ensuremath{{\boldsymbol #1}}} % for vectors
\newcommand{\gv}[1]{\ensuremath{\mbox{\boldmath$ #1 $}}}

\newcommand{\nn}[0]{\ensuremath{\boldsymbol{n}}}
\newcommand{\ww}[0]{\ensuremath{\boldsymbol{w}}}
\newcommand{\uu}[0]{\ensuremath{\boldsymbol{u}}}

\newcommand{\XT}[1]{\ensuremath{X_T^{#1}(\Omega)}}

\newcommand{\Sx}[0]{\ensuremath{\v{S}_X}}
\newcommand{\Sy}[0]{\ensuremath{\v{S}_Y}}

\newcommand{\BB}[0]{{\ensuremath{\boldsymbol{\cal B}}}}
\renewcommand{\AA}[0]{\ensuremath{\boldsymbol{A}}}

\renewcommand{\div}[1]{\nabla \cdot #1} % for divergence
\renewcommand{\curl}[1]{\nabla \times #1} % for curl

\newcommand{\FF}[0]{\ensuremath{\v{F}}} % for unit vector

\title{A variational principle for three-dimensional water waves over Beltrami flows}

\author{
	E.~Lokharu\thanks{Centre for Mathematical Sciences, Lund University, PO Box 118, 22100 Lund, Sweden}
	\and E.~Wahl\'en\footnotemark[1]
}

\begin{document}

	\maketitle
	\abstract{ We consider steady three-dimensional gravity-capillary water waves with vorticity propagating on water of finite depth. We prove a variational principle for doubly periodic waves with relative velocities given by Beltrami vector fields, under general assumptions on the wave profile.}

\section{Introduction}

This paper is concerned with  three-dimensional steady water waves  driven by gravity and surface tension. Almost all previous investigations of such waves have worked under the assumption of irrotational flow. In this paper, on the other hand, we allow for non-zero vorticity. This could be important for modelling three-dimensional interactions of waves with non-uniform currents. While our study is limited to Beltrami fields, even this particular case is a step forward compared to the previous state of knowledge.
 The fluid domain $\Omega \subset \RR^3$ is assumed to be an open, simply connected set, bounded from below by a rigid flat bottom $\bottom=\{z=-d\}$ and from above by a free surface $\surf$, separating the fluid from the air.  Let $\uu\colon\overline{\Omega} \to \RR^3$ be the (relative) velocity field and $p\colon \overline{\Omega} \to \RR$ the pressure. In a moving frame of reference, the fluid motion is governed by the steady Euler equations
 \begin{subequations}
\begin{align}
& (\uu \cdot \nabla) \uu = - \nabla p - g \v{e}_3 && \text{in} \ \ \Omega, \label{Euler1} \\
& \div{\uu} = 0 && \text{in} \ \ \Omega, \label{Euler2} \\
\intertext{with kinematic boundary condition on the top and bottom boundaries}
& \uu \cdot \v{n} = 0 && \text{on} \ \ \partial \Omega, \nonumber \\
\intertext{and dynamic boundary condition on the free surface}
& p=-2\sigma K_M  && \text{on} \ \ \partial\Omega^t.  \nonumber
\end{align}
\end{subequations}
Here $\v{e}_3 = (0,0,1)$ and $K_M$ is the mean curvature of the free surface defined by
$2K_M=-\div  \v{n}$ where $\v{n}$ is the unit outward normal, while $\sigma > 0$ is the surface tension coefficient.

 In the classical situation the free surface is given by the graph of a function, which excludes overhanging wave profiles. In this paper we consider a more general geometry as e.g~in \cite{AR, BB, BT}, by defining the fluid domain as
\[
\Omega = \v{F}(D), \ \ D = \RR^2 \times (-d,0)
\]
for some $d>0$ and a map $\v{F} \colon \RR^3 \to \RR^3$ satisfying the following conditions:
\begin{itemize}
	\item[(F1)] $\v{F}\colon\RR^3 \to \RR^3$ is a diffeomorphism with bounded partial derivatives and $\det{\textrm{D}\v{F}} > 0$;
	\item[(F2)] $\v{F}(X,Y,-d) = (X,Y,-d)$ for all $X,Y \in \RR$.
\end{itemize}
Thus, the free surface is given by
\[
\partial \Omega^t = \{\FF(X,Y,0): (X,Y) \in \RR^2\},
\]
while the flat bottom is
\[
\partial \Omega^b = \{ (X,Y,-d): (X,Y) \in \RR^2 \}.
\]
In what follows we will use the notation
\[
\v{S}(X,Y) = \FF(X,Y,0)
\]
for the surface parametrization. Our assumptions allow overhang but exclude self-intersection and they imply that $\partial \Omega^t$ lies above $\partial \Omega^b$

The set of all $\v{F} \in C_{loc}^{3,\gamma}(\RR^3;\RR^3)$ satisfying (F1) and (F2)  will be denoted by $M$.
Here $C^{k, \gamma}(U)$, with $k\in \mathbb N_0=\{0,1,2,\ldots\}$ and $\gamma\in (0,1)$, denotes the class of $k$ times continuously differentiable functions whose partial derivatives of order less than or equal to $k$ are bounded and uniformly $\gamma$-H\"older continuous. The notation $C_{loc}^{k, \gamma}(U)$ will be used for the space of functions satisfying this condition in a neighbourhood of each point in $U$. Throughout the rest of the paper, we will continuously extend functions in $C^{k,\gamma}(U)$ to the boundary of $U$ without explicit mention.

We will consider doubly-periodic waves as follows.
Let
\begin{equation} \label{lattice} \nonumber
\Lambda = \{ \gv{\lambda} = l \gv{\lambda}_1 + j \gv{\lambda}_2 : l,j \in \mathbb{Z} \},
\end{equation}
be a two-dimensional latticegenerated by two linearly independent vectors $\gv{\lambda}_1,\gv{\lambda}_2 \in \RR^2$, and let
\[
B = \{a_1 \gv{\lambda}_1 + a_2 \gv{\lambda}_2 : a_1,a_2 \in [0,1] \}
\]
be a two-dimensional periodic cell in the lattice. We will assume that $\v{F}(\v{x})-\v{x}$ is periodic with respect to the lattice, so that
\[
\v{F}(\v{X}'+\gv{\lambda},Z)=\v{F}(\v{X}',Z) +(\gv{\lambda},0), \ \ \v{S}(\v{X}'+\gv{\lambda})
= \v{S}(\v{X}') + (\gv{\lambda},0)
\]
for all $\gv{\lambda} \in \Lambda$ and $\v{X}' = (X,Y) \in \RR^2$, and we denote by $M_{per}$ the set of all $\v{F} \in M$ satisfying this property.
We will consider periodic solutions, meaning that
\[
\uu(\v{x}'+\gv{\lambda},z)=\uu(\v{x}',z),  \quad p(\v{x}'+\gv{\lambda},z)=p(\v{x}', z)
\]
for all $\gv{\lambda}\in \Lambda$.

In the irrotational case, when $\curl{\uu} = 0$ everywhere in $\Omega$, there are several existence results for different types of three-dimensional waves, including doubly-periodic waves, fully localized solitary waves and waves with a solitary-wave profile in one horizontal direction and periodic or quasi-periodic profile in another
(see e.g.~\cite{BGSW, BGW, CN, GH, GM, GS, IP1, IP2, RS} and references therein).

On the the other hand, the existence of genuinely three-dimensional water waves with vorticity is completely open, except for a non-existence result for water waves with constant vorticity \cite{Wahlen14}.
Even in the absence of a free surface, the literature concerning steady rotational flows with vorticity is pretty scarce. There are only a handful of general existence results for steady flows with vorticity in fixed domains \cite{A, BW, TX}. However, the special case when the velocity and vorticity fields are collinear, that is,
\begin{align}
\curl{\uu} = \alpha \uu \quad \text{in} \ \ \Omega, \label{Beltrami} \nonumber
\end{align}
for some scalar function $\alpha$, has received more attention. Such vector fields are known as Beltrami vector fields or force-free fields and are well-known in solar and plasma physics (see e.g.~\cite{Freidberg, Priest}).
Any divergence-free Beltrami field generates a  solution to the Euler equation \eqref{Euler1} with  pressure given by
\[
p = C - \frac{|\uu|^2}{2} - g z.
\]
In general, condition \eqref{Euler2} is satisfied if $\alpha$ is constant along the streamlines of $\uu$. In this paper we will however concentrate on the case when $\alpha$ is constant throughout the whole fluid. Such fields are often called strong Beltrami fields or linear force-free fields.
The theory for strong Beltrami fields is much more developed than for Beltrami fields with variable $\alpha$ (see e.g.~the discussions in \cite{KNW} and \cite{BA}) and in fact an obstruction to finding fields with variable $\alpha$ was recently discovered in \cite{EPS3}.
In the following, we shall simply take Beltrami fields to mean strong Beltrami fields.
Beltrami fields are intimately connected with chaotic motion, the famous ABC flow \cite{AK} being a classical example. In \cite{EPS1} it was shown that any locally finite link can be obtained as a collection of streamlines of some Beltrami field and in \cite{EPS2} a similar result was shown for vortex tubes. Note that linear dependence of $\uu$ and $\curl{\uu}$ is in some sense necessary for chaotic behaviour by a theorem of Arnold \cite{AK}.

For a Beltrami field the governing equations are
\begin{subequations}
	\begin{align}
	&\div{\uu} = 0  \ \ && \text{in} \ \ \Omega, \label{B1} \\
	&\curl{\uu} = \alpha \uu  \ \ && \text{in} \ \ \Omega, \label{B2} \\
	&\uu \cdot \v{n} = 0  \ \  && \text{on} \ \ \partial\Omega,  \label{B3} \\
&\frac 1 2 |\uu|^2 + g z - 2\sigma K_M  = \text{const} && \text{on} \ \ \partial\Omega^t. \label{B4}
	\end{align}
\end{subequations}
The aim of the present paper is to find a variational formulation (Theorem \ref{MainTh}) for this problem in the periodic case. Classical and modern variational formulations \cite{CSS, L, Wahlen07} have proved useful in a variety of existence and stability theories for periodic and solitary travelling water waves. This includes two- and three-dimensional waves in the irrotational setting (see the references mentioned above)  as well as two-dimensional waves with vorticity \cite{BB, BT, GrovesWahlen07, GW1}. It is therefore natural to expect that a variational principle for three-dimensional waves over Beltrami flows could be useful.
In the absence of a free boundary, there is a classical variational formulation by Woltjer \cite{W1} which was further developed by Laurence \& Avellaneda \cite{AL} (see also the related formulation by Chandrasekhar \& Woltjer \cite{CW}). It states that Beltrami fields are critical points of the energy subject to the constraint of fixed helicity.
The presence of the free boundary requires some nontrivial modifications of this formulation, as does the different geometric setting. The first step is to construct vector potentials satisfying certain boundary conditions (Theorem \ref{VectPot}). The variational principle (Theorem \ref{MainTh}) is then formulated in terms of such potentials.
In our presentation we have striven for a balance between rigor and simplicity. The variational formulation is presented in a mathematically rigorous fashion in terms of certain function spaces, but we have tried not to overemphasize technical details. The choice of variational formulation in  Theorem \ref{MainTh} is certainly not unique. We give some comments about this after the proof of the theorem,  which could be useful for a variational existence theory. In addition, it would also have been possible to use other function spaces, such as Sobolev spaces.

\section{Vector potentials}

A vector potential of $\uu$ is a vector field $\AA$ such that
\[
\curl{\AA} = \uu.
\]
Such a potential is not unique since we can add to it the gradient of any smooth function $\phi$.
In order to derive a variational principle for Beltrami flows, we need to examine the structure of vector potentials for periodic  vector fields satisfying $\uu\cdot \v{n}=0$ on $\partial \Omega$.
For $k \in \mathbb{N}_0$ and $\gamma \in (0,1)$ we put
\begin{align*}
& X^k(\Omega) = C^{k,\gamma}_{per}(\Omega; \RR^3), \\
& X^k_N(\Omega) = \{\v{v}\in X^k(\Omega): \v{v} \times \v{n} = 0 \  \text{on} \ \partial \Omega \}, \\
& X^k_T(\Omega) = \{\v{v}\in X^k(\Omega): \v{v} \cdot \v{n} = 0 \  \text{on} \  \partial \Omega \}.
\end{align*}
The subscript \textit{per} stands for periodicity with respect to the lattice $\Lambda$.
We will also need the divergence-free analogues of the above spaces, defined by
\[
Y^k(\Omega) = \{ \v{v}\in X^k(\Omega): \div{\v{v}} = 0 \  \text{in} \  \Omega \}
\]
and
\begin{align*}
& Y^k_N = Y^k(\Omega)\cap X^k_N(\Omega), \\
& Y^k_T =Y^k(\Omega)\cap X^k_T(\Omega).
\end{align*}
 When $\uu \in Y^k_T$ has zero fluxes, there is a unique vector potential from $Y^{k+1}_N$, $k \geq 2$. This is no longer true for non-zero fluxes. However, one can prove the following statement.

\begin{theorem}\label{VectPot} For any vector field $\uu \in Y_T^1(\Omega)$ there exists a vector potential $\AA \in Y^{2}(\Omega)$ such that
\begin{align}
& \AA \times \nn = 0 \ \ \text{on} \ \ \partial\Omega^t, \label{VP1} \\
& \AA \times \nn = (m_1,m_2,0) \ \ \text{on} \ \ \partial\Omega^b, \label{VP2}
\end{align}
for some constants $m_1,m_2 \in \RR$ determined by $\uu$. On the other hand, if $\AA \in Y^{2}(\Omega)$ satisfies \eqref{VP1} and \eqref{VP2}, then $\curl{\AA} \in Y_T^1(\Omega)$.
\end{theorem}

Let us explain the connection between the constants $m_1,m_2$ and the fluxes $a_1,a_2$ corresponding to the vertical sides $\Sigma_1,\Sigma_2$ of a basic periodic cell of the fluid domain that are parallel to the lattice vectors $\gv{\lambda}_1$ and $\gv{\lambda}_2$ respectively. Using  Stokes theorem, we find that
\[
a_j = \int_{\Sigma_j} \curl{\AA} \cdot \D\v{S} = \oint_{\partial \Sigma_j} \AA \cdot \D\v{L} = (m_2,-m_1) \cdot \gv{\lambda}_j, \ \ j=1,2.
\]
Here the vertical components of the contour integral cancel due to the periodicity and the top part is zero because of \eqref{VP1}.

The proof of the theorem relies on the following regularity result for the Biot-Savart integral
\[
\BS_\Omega(\uu) \coloneqq  \frac{1}{4\pi} \int_{\Omega} \uu(\v{y}) \times \frac{\v{x} - \v{y}}{|\v{x} - \v{y}|^3}\, \D V.
\]

\begin{lemma}\label{Biot-Savart} Let $\Omega$ be a bounded domain with $C^{3,\gamma}$-smooth boundary and let $\uu \in C^{1,\gamma}(\Omega)$. Then $\BS_{\Omega}(\uu) \in C^{2,\gamma}(\Omega)$ and
	\[
	\|\BS_\Omega(\uu)\|_{C^{2,\gamma}(\Omega)} \leq C \|\uu\|_{C^{1,\gamma}(\Omega)},
	\]
	where the constant $C=C(\Omega,\gamma)$ depends only on the domain $\Omega$ and $\gamma$.
\end{lemma}

\begin{proof} Note that it is enough to consider the scalar operators
	\[
	B_m(u)(\v{x}) = \frac{1}{4\pi}\int_{\Omega} u(\v{y}) \frac{x_m-y_m}{|\v{x}-\v{y}|^3}\, \D V,
	\]
	which are the partial derivatives of the Newtonian potential
	\[
	I(u)(\v{x}) =\frac{1}{4\pi} \int_{\Omega} u(\v{y}) \frac{1}{|\v{x}-\v{y}|}\, \D V;
	\]
	see e.g.~\cite[Lemma 4.1]{GT}.
	It is well known that $I\colon C^{1,\gamma}(\Omega) \to C^{3,\gamma}(\Omega)$, which implies that
	\[
	B_m\colon C^{1,\gamma}(\Omega) \to C^{2,\gamma}(\Omega).
	\]
	This finishes the proof of the lemma.
\end{proof}

\begin{proof}[Proof of Theorem \ref{VectPot}] Let us define
\[
\Lambda_{lj} = \{(a_1+l) \gv{\lambda}_1 + (a_2+j) \gv{\lambda}_2 : a_1,a_2 \in (0,1) \}, \ D_{lj} = \Lambda_{lj} \times (-d,0)
\]	
and
\[
\Omega_{lj} = \v{F}(D_{lj}), \ \ l,j \in \mathbb{Z},
\]	
which splits the domain $\Omega$ into simple periodic cells. Furthermore, let us consider a $C^{3,\gamma}$-domain $\tilde{\Omega} \supset \Omega_{00}$ such that $\tilde{\Omega} \subset \cup_{|l|,|j| \leq 1} \overline{\Omega_{lj}}$. Thus, $\tilde{\Omega}$ intersects only eight neighbouring cells. Now let $\phi_{lj} \in C^{\infty}$ be a partition of unity on $\Omega$ such that (i) $\phi_{lj} = 1$ on $\Omega_{lj}$; (ii)  $\phi_{lj} = 0$ on $\Omega \setminus \tilde{\Omega}_{lj}$, where $\tilde{\Omega}_{lj}=\tilde{\Omega}+(l \gv{\lambda}_1 + j \gv{\lambda}_2,0)$; (iii) $\phi_{lj}(\v{x})=\phi_{00}(\v{x}-l\gv{\lambda}_1-j\gv{\lambda}_2)$, $l,j \in \mathbb{Z}$. Then, for a given vector field $\uu \in \XT{1}$, we let
\[
\BB_{lj}(\v{x})=\BS_{\tilde{\Omega}_{lj}}(\phi_{lj}\uu)(\v{x})=\frac{1}{4\pi} \int_{\tilde{\Omega}_{lj}} \phi_{lj}(\v{y}) \uu(\v{y}) \times \frac{\v{x} - \v{y}}{|\v{x} - \v{y}|^3}\, \D V,
\]
which is the Biot-Savart integral of $\phi_{lj}\uu$. The latter integral converges for all $\v{x} \in \RR^3$ and $\BB_{lj} \in C^{2,\gamma}(\tilde{\Omega}_{lj})$, as can be seen using Lemma \ref{Biot-Savart}.

Because the vector fields $\BB_{lj}$ are not periodic in general, we define
\begin{equation} \label{BBdef}
\BB(\v{x}) = p.v. \sum_{lj} \BB_{lj}(\v{x}) \coloneqq \frac{1}{2}\sum_{l,j \in {\mathbb Z}} [\BB_{lj}(\v{x}) + \BB_{(-l)(-j)}(\v{x})].
\end{equation}
Now because
\[
|\BB_{lj}(\v{x}) + \BB_{(-l)(-j)}(\v{x})| \leq C(\v{x}) \frac{1}{1 + |l|^3 + |j|^3}, \ \ j,l \in {\mathbb{Z}},
\]
for some bounded function $C(\v{x})$, the series in \eqref{BBdef} converges uniformly and so $\BB$ is well defined and continuous everywhere.
Furthermore, since
\[
|\textrm{D}^k\BB_{lj}(\v{x})| \le C(\v{x}) \frac{1}{1+|l|^{2+k}+|j|^{2+k}}
\]
for all $k\in \mathbb N$ and $\v{x}\in \Omega \setminus \overline{ \tilde \Omega_{lj}}$, we obtain that
$\BB$ is periodic, $\BB \in C^{2,\gamma}(\Omega; \RR^3)$ and $\curl{\BB} = \uu$. Let us prove that there exists a function $g \in C^{3, \gamma}_{loc}(\RR^3)$ such that $\BB-\nabla g$ satisfies the boundary conditions. We let
\[
\BB^t = \BB|_{\partial\Omega^t} - (\BB|_{\partial\Omega^t} \cdot \nn) \nn
\]
be the tangential part of the field $\BB$ on the top boundary.  Consider the tangent vectors
$\v{S}_X$ and $\v{S}_Y$ and define a two-dimensional vector field $B^* = (B^*_1,B^*_2) \in C^{2,\gamma}_{per}(\RR^2; \RR^2)$ by
\[
\begin{split}
& B_1^*(X,Y) = \BB^t(\v{S}(X,Y)) \cdot \v{S}_X(X,Y)= \BB(\v{S}(X,Y)) \cdot \v{S}_X(X,Y), \\
& B_2^*(X,Y) = \BB^t(\v{S}(X,Y)) \cdot \v{S}_Y(X,Y)= \BB(\v{S}(X,Y)) \cdot \v{S}_Y(X,Y).
\end{split}
\]
We claim that the field $B^*$ is conservative, that is there exists a function $f \in C^{3,\gamma}_{loc}(\RR^2)$ (not necessarily periodic) such that $B^* = \nabla f$.
Indeed, this follows from the relation
\[
\v{n}\cdot \curl \BB= \frac{1}{|\v{S}_X\times \v{S}_Y|}((B_2^*)_X-(B_1^*)_Y)
\]
(see \cite[\S 97]{Brand}) and the fact that $\uu\cdot \nn=0$ on $\partial \Omega$.
Because $B^*$ is periodic, we necessarily have
\[
B^* = \nabla f_0 + (a_1,a_2),
\]
where $f_0 \in C^{3, \gamma}(\RR^2)$ is periodic and $a_1,a_2$ are constants.

A similar argument is valid for the bottom boundary. In this case the tangential vectors are
\[
\widetilde{\ww}_1 = (1,0,0), \ \ \widetilde{\ww}_2 = (0,1,0)
\]
and the corresponding two-dimensional field is given by
\[
\widetilde{B}_j^*(x,y) = \BB(x,y,-d) \cdot \widetilde{\ww}_j(x,y), \ \ j=1,2.
\]
Just as before, we obtain
\[
\widetilde{B}^* = \nabla \widetilde{f}_0 + (\widetilde{a}_1,\widetilde{a}_2)
\]
for some periodic function $\widetilde{f}_0$ and constants $\widetilde{a}_1,\widetilde{a}_2$.

In order to eliminate the tangential periodic part of $\BB$ on the boundary, we solve the Dirichlet problem
\begin{align*}
& \Delta \phi = 0 \ \ \text{in} \ \ \Omega, \\
& \phi = f_0 \circ \v{F}^{-1} \ \ \text{on} \ \ \partial\Omega^t, \\
& \phi = \widetilde{f}_0 \ \ \text{on} \ \ \partial\Omega^b.
\end{align*}
Because both functions $f_0$ and $\widetilde{f}_0$ are periodic, there is a unique periodic solution $\phi \in C^{3, \gamma}(\Omega)$.
We also define $\Phi(\v{x})=\v{F}^{-1}(\v{x})$ and let $\Phi_\text{per}(\v{x})=\Phi(\v{x})-\v{x}$ be its periodic part. Letting $\phi_j \in C^{3, \gamma}(\Omega)$, $j=1,2$, be the unique solution of the Dirichlet problem
\begin{align*}
& \Delta \phi_j = 0 \ \ \text{in} \ \ \Omega, \\
& \phi_j = \Phi_{\text{per},j} \ \ \text{on} \ \ \partial\Omega^t, \\
& \phi_j = \Phi_{\text{per},j}\ \ \text{on} \ \ \partial\Omega^b,
\end{align*}
we find that $\nabla(\phi_1+x) \cdot \v{S}_X=\nabla \Phi_1 \cdot \v{S}_X=1$ on $\partial \Omega^t$.
Similarly, $\nabla (\phi_2+y)\cdot \v{S}_Y=1$, while $\nabla (\phi_1+x)\cdot \v{S}_Y=\nabla (\phi_2+y)\cdot \v{S}_X=0$ on $\partial \Omega^t$, and $\nabla(\phi_1+x) \cdot \widetilde{\ww}_1=\nabla(\phi_2+y) \cdot \widetilde{\ww}_2=1$, $\nabla(\phi_1+x) \cdot \widetilde{\ww}_2=\nabla(\phi_2+y) \cdot \widetilde{\ww}_1=0$ on $\partial \Omega^b$ (note that $\Phi_{\text{per},j}=0$ on $\partial\Omega^b$ for $j=1,2$). Now we put
\[
\AA = \BB - \nabla (\phi + a_1 (\phi_1+x) + a_2 (\phi_2+y)).
\]
A direct calculation shows that
\[
[ \AA - (\AA \cdot \nn) \nn ] \cdot \v{S}_X = [ \AA - (\AA \cdot \nn) \nn ] \cdot \v{S}_Y = 0
\]
on $\partial\Omega^t$, so that
\[
\AA|_{\partial\Omega^t} \times \nn = 0.
\]
On the other hand, we get
\[
 \AA - (\AA \cdot \nn) \nn  = (\widetilde{a}_1 - a_1, \widetilde{a}_2 - a_2, 0)
\]
on $\partial\Omega^b$. Thus,
\[
\AA|_{\partial\Omega^b} \times \nn = (a_2-\widetilde{a}_2, a_1 - \widetilde{a}_1,0).
\]
This finishes the proof of the theorem.
\end{proof}

\section{Variational principle}

In this section we will formulate and prove a variational principle for Beltrami vector fields with a free surface. The principle is formulated in terms of the vector potential $\AA$ introduced in the previous section and the $C^{3,\gamma}$ domain $\Omega$. A key difficulty is that the surface of the domain is not fixed and is a part of the variation. The admissible domains will be parametrized by maps $\v{F}\in M_{per}$. One of the issues is that the vector potential $\AA$ depends on $\Omega$ through its domain of definition. This can be solved by extending $\AA$ to the whole of $\RR^3$. However, $\AA$ and $\Omega$ are still coupled through the boundary condition imposed on the free surface. Thus, the proper way to think about the domain of the involved functionals is as a submanifold of $C^{2,\gamma}(\RR^3;\RR^3)\times C_{loc}^{3,\gamma}(\RR^3; \RR^3)$. Rather than making this approach completely rigorous, we shall simply consider critical points along admissible families of curves. After presenting the the theorem and its proof, we will discuss some alternative perspectives on the variational formulation which might be useful for further studies.

Let $\Omega_{00}=\v{F}(D_{00})$ be a single cell of the periodic domain defined in Section 2 for which the top boundary is given by
\[
\partial \Omega_{00}^t = \{\v{S}(\v{X}'): \v{X}'=(X,Y) \in \Lambda_{00} \}.
\]
Let us consider the  functionals
\begin{equation} \label{en}
{\cal E}(\AA, \Omega) = \int_{\Omega_{00}} \left[ |\curl{\AA}|^2 - 2g z \right] \ \, \D V - \sigma \int_{\partial \Omega_{00}^t} \, \D S, \nonumber
\end{equation}
\begin{equation} \label{const1}
{\cal K}(\AA, \Omega) = \int_{\Omega_{00}} [\AA \cdot (\curl \AA)] \ \, \D V \nonumber
\end{equation}
and
\begin{equation} \label{const2}
{\cal M}(\AA, \Omega) = \int_{\Omega_{00}} 1 \ \, \D V, \nonumber
\end{equation}
where $\AA \in C^{2,\gamma}(\RR^3;\RR^3)$  satisfies
\begin{align}
 \AA \times \v{n} = 0 \ \ & \text{on } \surf, \label{cond} \\
 \AA \times \v{n} = (m_1,m_2,0) \ \ & \text{on }  \bottom \label{condbot}
\end{align}
for some fixed constants $m_1,m_2 \in \RR$. We look for critical points of the functional
\[
{\cal J}= {\cal E} - \alpha {\cal K} - \mu {\cal M}
\]
where $\alpha$ and $\mu$ are fixed constants. This is, at least formally, equivalent to considering
critical points of $\mathcal E$ subject to the constraints of fixed $\mathcal K$ and $\mathcal M$.
When taking variations of the functionals, we consider a family of domains $\Omega(t)=\v{F}(t)(D)$, $t\in (-\delta,\delta)$, where $\v{F}(t) \in M_{per}$ is family of domain parametrizations which is continuously differentiable in $t$ in the $C_{loc}^{3,\gamma}$ topology. Note that we have written $\Omega$ rather than $\v{F}$ in the arguments
of the functionals $\mathcal E$, $\mathcal K$ and $\mathcal M$ to emphasize that they only depend on $\Omega=\v{F}(D)$ and not on the specific parametrization. However, when taking derivatives later, we will write $\delta \FF$ to indicate the direction in which we differentiate.
We also consider a continuous family of vector fields $\AA(t) \in C_{per}^{2,\gamma}(\RR^3; \RR^3)$, $t\in (-\delta,\delta)$, such that $t \mapsto \AA(t)$ is differentiable when considered as a map from $(-\delta,\delta)$ into
$C_{per}^{1,\gamma}(\RR^3; \RR^3)$. The reason for only assuming differentiability with respect to the  $C^{1,\gamma}$ topology and not the $C^{2,\gamma}$ topology is that we use compositions to construct suitable curves and that this leads to a loss of derivatives.
The vector fields $\AA(t)$ are assumed to satisfy conditions \eqref{cond} and \eqref{condbot}.
Such curves will be called {\em admissible} as will the corresponding variations
\[
\delta \FF=\frac{d}{dt} \FF(t)\Big|_{t=0}, \qquad \delta \AA=\frac{d}{dt} \AA(t)\Big|_{t=0}.
\]
It follows from the admissibility conditions that
\begin{equation} \label{var3}
\delta \v{F}(X,Y, -d)=0, \nonumber
\end{equation}
and
\begin{equation} \label{var4}
\delta \v{F}(\v{X}'+\gv{\lambda}, Z)=\delta \v{F}(\v{X}', Z), \qquad \gv{\lambda}\in \Lambda, \nonumber
\end{equation}
while

\begin{equation} \label{var5}
 \delta \AA \times \v{n}  + (\textrm{D} \AA\, \delta \v{F} \circ \v{F}^{-1}) \times \v{n}
+\AA \times \frac{ \v{S}_X \times \delta \v{S}_Y + \delta \v{S}_X \times \v{S}_Y}{|\v{S}_X \times \v{S}_Y|}\circ \v{F}^{-1} = 0
\end{equation}
along the top boundary $\surf$ and
\begin{equation} \label{var5bot}
\delta \AA \times \v{n} = 0
\end{equation}
on the bottom $\bottom$, where $\v{F}=\v{F}(0)$, $\AA=\AA(0)$ etc. Note that \eqref{var5} gives a relation between $\delta \AA$ and $\delta \v{F}$, so that the variations are not independent.

Our main theorem is the following variational principle.
\begin{theorem} \label{MainTh} $(\AA,\Omega)$ is a critical point of the functional $\mathcal J$ if and only if $\AA$ is a vector potential of a steady Beltrami flow $\uu = \curl{\AA}$ in $\Omega$ satisfying
the boundary conditions \eqref{B3} and \eqref{B4} provided the variations are taken among admissible curves.
\end{theorem}

For the proof we will need the following technical lemma.

\begin{lemma}
\label{Teq}
Let $\v{F}\in M_{per}$ and $\AA \in C^{2,\gamma}(\RR^3)$, satisfying \eqref{cond} and \eqref{condbot}, be given.
The set $\{\delta \v{F}|_{Z=0}\cdot \v{n}\circ \v{S} \colon  (\delta \AA, \delta \v{F}) \text{ is admissible} \}$ is dense in $C^{2, \gamma}_{per}(\RR^2)$.
\end{lemma}

\begin{proof}
Let $\widehat{\delta \eta} \in C^{2, \gamma}_{per}(\RR^2)$ be given and extend it to a function
$\widehat{\delta \eta}\in C^{2, \gamma}_{per}(\RR^3)$ with $\widehat{\delta \eta}=0$ for $|Z|\ge d/2$.
The function
\[
\delta \v{F}=\widehat{\delta \eta}\, \frac{\v{F}_{X}\times \v{F}_{Y}}{|\v{F}_{X}\times \v{F}_{Y}|}
\]
then clearly satisfies
\[
\delta \v{F}|_{Z=0}\cdot \v{n}\circ \v{S} =\widehat{\delta \eta}.
\]
However, $\delta \v{F}$ is only in $C^{2,\gamma}$. Convolving $\delta \v{F}$ with a smooth mollifier which is $\v{X}'$-periodic and has compact support in $Z$, we obtain a function $\delta \v{F}_\varepsilon \in C_{per}^\infty(\RR^3; \RR^3)$ with compact support in $Z$ and
$\|\delta \v{F}_\varepsilon|_{Z=0}\cdot \v{n}\circ \v{S}-\widehat{\delta \eta}\|_{C^{2,\gamma}}<\varepsilon$. We set $\v{F}(t)=\v{F} +t\delta \v{F}_\varepsilon$ and notice that $\v{F}(t)$ is a diffeomorphism from $\RR^3$ to itself for sufficiently small $t$.

Next, for a given vector field $\op{\AA} = (\hat{A}^{(1)},\hat{A}^{(2)},\hat{A}^{(3)}) \in C_{per}^{2,\gamma}(\RR^3;\RR^3)$ and $\FF \in M_{per}$, we define a vector field $T_{\v{F}} \hat{\AA} \in C_{per}^{2,\gamma}(\RR^3; \RR^3)$ by
\begin{equation} \label{vfcorr}
T_{\v{F}} \op{\AA}(\v{x}) = \hat{A}^{(1)}(\v{X}) \FF_X(\v{X}) + \hat{A}^{(2)}(\v{X}) \FF_Y(\v{X}) + \hat{A}^{(3)}(\v{X}) \FF_X(\v{X}) \times  \FF_Y(\v{X}), \nonumber
\end{equation}
where $\v{X} = \FF^{-1}(\v{x})$. The map
\[
T_\v{F} \colon C_{per}^{2,\gamma}(\RR^3;\RR^3) \to C_{per}^{2,\gamma}(\RR^3; \RR^3)
\]
is a linear homeomorphism of Banach spaces. Furthermore, $\op{\AA}$
satisfies the boundary conditions
\begin{align}
&(\hat A_1, \hat A_2)=(0,0) & & \text{on } Z=0, \label{condhat} \\
&(\hat  A_1, \hat A_2)=  (m_1,m_2,0) && \text{on } Z=-d \label{condbothat}
\end{align}
if and only if $T_\v{F}\op{\AA}$ satisfies \eqref{cond} and \eqref{condbot}.
Clearly,  $T_\v{F} \in C^\infty(C_{per}^{2,\gamma}(\RR^3;\RR^3); C_{per}^{2,\gamma}(\RR^3; \RR^3))$,  while
\[
(\op{\AA},\v{F})\mapsto T_{\v{F}} \op{\AA}
\]
belongs to $C^1(C_{per}^{2,\gamma}(\RR^3;\RR^3)\times M_{per}; C_{per}^{1, \gamma}(\RR^3; \RR^3))$.
Setting
\[
\AA(t)=T_{\v{F}(t)} T_{\v{F}}^{-1} \AA,
\]
we find that $(\AA(t), \v{F}(t))$ is an admissible curve.
\end{proof}

\begin{proof}[Proof of Theorem \ref{MainTh}] A critical point is subject to the Euler-Lagrange equation
\begin{equation} \label{EL}
\delta{\cal J}(\AA, \Omega)[\delta \AA,\delta \v{F}] = (\delta {\cal E} - \alpha \delta {\cal K} - \mu \delta {\cal M})(\AA, \Omega)[\delta \AA,\delta \v{F}] = 0,
\end{equation}
where $\alpha, \mu \in \RR$ are Lagrange multipliers.
This equation is valid for all admissible $(\delta \AA, \delta\v{F})=(\AA'(0), \v{F}'(0))$.
First we will show that \eqref{EL} implies \eqref{B2} for $\uu \coloneqq \curl{\AA}$. Then using this fact we will derive the boundary equation \eqref{B4}.

Let us calculate all the variations in \eqref{EL}. A direct calculation gives
\begin{equation} \label{varE} \nonumber
\begin{split}
\delta {\cal E} = 2 \int_{\dom_{00}} (\curl{\AA}) \cdot (\curl{ \delta \AA}) \, \D V & + \int_{\surf_{00}} (|\curl \AA|^2 - 2gz ) \delta \eta \, \D S + \sigma \int_{\surf_{00}} 2 K_M \delta \eta \, \D S,
\end{split}
\end{equation}
where $\delta \eta=\widetilde{\delta \v{S}}\cdot \v{n}$ and
we have used the notation $\widetilde{\delta \v{S}} = \delta \v{S} \circ \v{F}^{-1}$ for convenience. This notation will also be used for other functions below.
An application of the divergence theorem to the first integral leads to
\[
\int_{\dom_{00}} (\curl{\AA}) \cdot (\curl{\delta \AA}) \, \D V = \int_{\dom_{00}} [\curl{(\curl{\AA})}] \cdot \delta \AA \, \D V + \int_{\partial \Omega_{00}} \delta \AA \times (\curl{\AA}) \cdot \D \v{S}.
\]
The boundary integral above equals
\begin{equation} \label{varEb} \nonumber
\begin{split}
\int_{\boundary_{00}} [\delta \AA \times (\curl{\AA})] \cdot \v{n}  \, \D S & = \int_{\bottom_{00}} [\delta \AA \times \v{n}] \cdot (\curl{\AA})  \, \D S  -\int_{\surf_{00}} [\delta \AA \times \v{n}] \cdot (\curl{\AA})  \, \D S.
\end{split}
\end{equation}
The boundary integral over $\bottom$ is zero in view of \eqref{var5bot}.
Similarly, we find
\begin{equation} \label{varK}
\begin{split}
 \delta {\cal K}  & = \int_{\dom_{00}} [ \delta \AA \cdot (\curl{\AA}) + \AA \cdot (\curl{\delta \AA}) ] \ \, \D V + \int_{\surf_{00}} [\AA \cdot (\curl{\AA})]\delta \eta  \, \D S  \\
 & =  2 \int_{\dom_{00}}  (\curl{\AA}) \cdot \delta \AA \ \, \D V +  \int_{\surf_{00}} [\AA \cdot (\curl{\AA})]\delta \eta  \, \D S,
\end{split}
\end{equation}
where the boundary terms arising through integration by parts vanish due to \eqref{cond} and \eqref{var5bot}.

Note that any pair $(\delta \AA , 0)$, with $\delta \AA \in X_N^2(\Omega)$ is admissible since we can choose $\v{F}(t)\equiv \FF$ and $\AA(t)=\AA+t \delta \AA$. Evaluating
$\delta{\cal J}$ for such variations we get
\begin{equation}
\label{interior variations}
\delta{\cal J}(\AA, \Omega)[\delta \AA,0] = 2 \int_{\Omega_{00}} \left[ \curl{(\curl \AA)} - \alpha \curl{\AA} \right] \cdot\delta \AA \, \D V = 0.
\end{equation}
This gives \eqref{B2} for $\uu = \curl \AA$. Furthermore, \eqref{cond} and \eqref{condbot} imply \eqref{B3}. Note that \eqref{B1} is valid for $\uu$ automatically. Thus, it is left to verify the boundary relation \eqref{B4}.

Taking into account \eqref{B3} and \eqref{cond} we find that surface integral in \eqref{varK} vanishes. Thus, we have
\[
\begin{split}
\delta{\cal J}&= -2\int_{\surf_{00}} [\delta \AA \times \v{n}] \cdot (\curl{\AA}) \ \, \D S \\ &\quad +  \int_{\surf_{00}} \left(|\curl \AA|^2 - 2gz + \sigma 2 K_M - \mu \right) \delta \eta \, \D S \\
& = 2I + J.
\end{split}
\]
Let us calculate $I$. For this purpose we use \eqref{var5} to write
\begin{equation} \label{varI} \nonumber
\begin{split}
I & = \int_{\surf_{00}} (\textrm{D}\AA\,  \widetilde{\delta \v{S}})\times \v{n} \cdot (\curl{\AA})  \, \D S \\
&\quad + \int_{\surf_{00}} \AA\times \frac{\widetilde{\v{S}_X} \times \widetilde{\delta \v{S}_Y} + \widetilde{\delta \v{S}_X} \times \widetilde{\v{S}_Y}}{|\widetilde{\v{S}_X} \times \widetilde{\v{S}_Y}|} \cdot (\curl{\AA})  \, \D S \\
& = I_1 + I_2.
\end{split}
\end{equation}
The integrals  $I_1$ and $I_2$ can be rewritten as
\[
\begin{split}
& I_1 = - \int_{\surf_{00}} \left( (\textrm{D}\AA\, \widetilde{ \delta \v{S}}) \times (\curl{\AA}) \right) \cdot \v{n}\,  \, \D S, \\
& I_2 = - \int_{\surf_{00}} (\AA \times (\curl{\AA})) \cdot \v{j}  \,\, \D S,
\end{split}
\]
where
\[
\v{j} =  \frac{\widetilde{\v{S}_X} \times \widetilde{\delta \v{S}_Y} + \widetilde{\delta \v{S}_X} \times \widetilde{\v{S}_Y}}{|\widetilde{\v{S}_X} \times \widetilde{\v{S}_Y}|} .
\]
Note that the normal component of $\v{j}$ does not contribute to $I_2$; therefore only the tangential component $\v{j}_{||}$ is of interest. Let us show that
\begin{equation}
\label{j tangential}
\v{j}_{||}= -\nabla_{||} \delta \eta+(\widetilde{\delta \v{S}}\cdot \v{n}_X)\v{a}+(\widetilde{\delta \v{S}} \cdot \v{n}_Y)\v{b},
\end{equation}
where
\[
 \nabla_{||} \delta \eta \coloneqq \nabla \delta \eta - \v{n} \cdot \nabla \delta \eta \v{n}
=\delta \eta_X \v{a}+\delta \eta_Y \v{b}
\]
is the surface gradient of $\delta \eta$, and
\[
\v{a}=\frac{\widetilde{\v{S}_Y}\times \v{n}}{|\widetilde{\v{S}_X}\times \widetilde{\v{S}_Y}|}, \quad
\v{b}=\frac{ \v{n}\times \widetilde{\v{S}_X}}{|\widetilde{\v{S}_X}\times \widetilde{\v{S}_Y}|}
\]
are the dual vectors of $\widetilde{\v{S}_X}$, $\widetilde{\v{S}_Y}$ in the tangent plane to $\surf_{00}$.
Here the notation $\delta \eta_X$ ($\delta \eta_Y$) should be interpreted as a directional derivative in the direction $\widetilde{\v{S}_X}$ ($\widetilde{\v{S}_Y}$).
Using the identities
\[
\begin{split}
& \Sx \cdot (\Sx \times \nn) = \Sy \cdot (\nn \times \Sy) = 0, \\
& \Sy \cdot (\Sx \times \nn) = \Sx \cdot (\nn \times \Sy) =  - |\Sx \times \Sy|,
\end{split}
\]
we can compute
\[
\begin{split}
&\widetilde{\Sx} \cdot \v{j}= -\widetilde{\delta \Sx}\cdot \v{n}, \\
&\widetilde{\Sy}  \cdot \v{j}=  -\widetilde{\delta \Sy} \cdot \v{n}.
\end{split}
\]
On the other hand, differentiating the identity $\delta \eta=\widetilde{\delta \v{S}}\cdot \v{n}$, we find that
\[
\begin{split}
&\delta \eta_X=\widetilde{\delta \Sx} \cdot \v{n} + \widetilde{\delta \v{S}} \cdot \v{n}_X,\\
&\delta \eta_Y=\widetilde{\delta \Sy} \cdot \v{n} + \widetilde{\delta \v{S}} \cdot \v{n}_Y.
\end{split}
\]
Combining these identities with the relation
\[
\v{j}_{||}=(\v{j} \cdot \widetilde{\v{S}_X})\v{a}+(\v{j} \cdot \widetilde{\v{S}_Y})\v{b}
\]
we obtain \eqref{j tangential}.
Now because the vector field $\AA \times (\curl{\AA})$ is tangent to the surface, we can integrate by parts (see \cite{Brand} for details about surface gradient and surface divergence operators), obtaining
\[
\begin{split}
 \int_{\surf_{00}} (\AA \times (\curl{\AA})) \cdot \nabla_{||}\delta \eta  \, \D S
 & = - \int_{\surf_{00}} \nabla_{||} \cdot [\AA \times (\curl{\AA})] \delta \eta  \, \D S \\
	 & = - \int_{\surf_{00}} \div  [\AA \times (\curl{\AA})] \delta \eta  \, \D S \\
	 &\quad + \int_{\surf_{00}} \frac{\partial [\AA \times (\curl{\AA})] }{\partial \nn} \cdot \nn \, \delta \eta  \, \D S.
\end{split}
\]
Using the identity
\[
\div (\AA \times \v{B}) = (\curl \AA) \cdot \v{B} - \AA \cdot (\curl \v{B})
\]
combined with \eqref{B2}, \eqref{B3} and \eqref{cond},
we obtain
\[
\int_{\surf_{00}} \div [\AA \times (\curl{\AA})] \delta \eta \, \D S = \int_{\surf_{00}} |\curl{\AA}|^2 \ \delta \eta  \, \D S,
\]
while
\[
\begin{split}
\int_{\surf_{00}} \frac{\partial [\AA \times (\curl{\AA})] }{\partial \nn} \cdot \nn \ \delta \eta  \, \D S & = \int_{\surf_{00}} \left[\frac{\partial \AA}{\partial \nn} \times (\curl{\AA}) \right] \cdot \nn \ \delta \eta  \, \D S\\
& \ \ \ + \int_{\surf_{00}} \left[ \AA \times \frac{\partial (\curl{\AA})}{\partial \nn} \right] \cdot \nn \ \delta \eta  \, \D S.
\end{split}
\]
The second integral here is zero because of \eqref{cond}.
It follows that
\[
\begin{split}
I_2=&- \int_{\surf_{00}} |\curl{\AA}|^2 \ \delta \eta  \, \D S
+\int_{\surf_{00}} \left[\frac{\partial \AA}{\partial \nn} \times (\curl{\AA}) \right] \cdot \nn \ \delta \eta  \, \D S\\
&- \int_{\surf_{00}} (\AA \times (\curl{\AA})) \cdot ((\widetilde{\delta \v{S}}\cdot \v{n}_X)\v{a}+ (\widetilde{\delta \v{S}}\cdot \v{n}_Y)\v{b})  \, \D S.
\end{split}
\]
We have
\[
\begin{split}
I_1=&- \int_{\surf_{00}} \left( \frac{\partial \AA}{\partial \v{n}} \times (\curl{\AA}) \right) \cdot \v{n} \, \delta \eta  \, \D S\\
&
-\int_{\surf_{00}} \left( (\textrm{D}\AA\,  \widetilde{\delta \v{S}}_{||}) \times (\curl{\AA}) \right) \cdot \v{n} \, \D S
\end{split}
\]
where $\widetilde{\delta \v{S}}_{||} =\widetilde{\delta \v{S}}-\delta \eta \v{n}$ is the tangential part of $\widetilde{\delta \v{S}}$.
Differentiating \eqref{cond} in a tangential direction $\v{v}$, we obtain
\[
(\textrm{D}\AA \, \v{v}) \times \v{n}=-\AA\times (\textrm{D}\v{n} \, \v{v}).
\]
We can therefore rewrite $I_1$ yet again as
\[
\begin{split}
I_1=& -\int_{\surf_{00}} \left( \frac{\partial \AA}{\partial \v{n}} \times (\curl{\AA}) \right) \cdot \v{n} \, \delta \eta \, \D S\\
&
+\int_{\surf_{00}} (\AA \times (\curl{\AA}))\cdot (\textrm{D} \v{n}\, \widetilde{\delta \v{S}}_{||}) \, \D S.
\end{split}
\]
Noting that $\AA\times (\curl{\AA})$ is tangent to $\partial \Omega_{00}^t$, we can rewrite the last integral
using the tangential part of $\textrm{D} \v{n}\, \widetilde{\delta \v{S}}_{||}$. By the symmetry of the shape operator, we find that
\[
\textrm{D} \v{n}\, \widetilde{\delta \v{S}}_{||} \cdot \widetilde{\Sx}=
 \widetilde{\delta \v{S}}_{||} \cdot \textrm{D} \v{n}\,\widetilde{\Sx}
 =\widetilde{\delta \v{S}} \cdot \v{n}_X
\]
and similarly
\[
\textrm{D} \v{n}\, \widetilde{\delta \v{S}}_{||} \cdot \widetilde{\Sy}
 =\widetilde{\delta \v{S}} \cdot \v{n}_Y.
\]
Hence,
\[
\begin{split}
I_1=& -\int_{\surf_{00}} \left( \frac{\partial \AA}{\partial \v{n}} \times (\curl{\AA}) \right) \cdot \v{n} \, \delta \eta \, \D S\\
&
+\int_{\surf_{00}} (\AA \times (\curl{\AA})) \cdot ((\widetilde{\delta \v{S}}\cdot \v{n}_X)\v{a}+ (\widetilde{\delta \v{S}}\cdot \v{n}_Y)\v{b})  \, \D S.
\end{split}
\]
 Finally, combining the calculations for $I_1$ and $I_2$, we conclude that
\[
\delta{\cal J}(\AA, \Omega_{00})[\delta \AA,\delta \eta] = - \int_{\surf_{00}} \left[ |\curl \AA|^2 + 2gz - 2\sigma  K_M + \mu \right] \delta \eta \, \D S = 0.
\]
Since we can take $\delta \eta$ to be an arbitrary smooth periodic function by Lemma \ref{Teq}, we recover the boundary condition \eqref{B4}.

Conversely, if $\AA$ is a vector potential of a steady Beltrami flow in $\Omega$ satisfying \eqref{B3} and \eqref{B4} one finds by the above formulas that $(\AA, \Omega)$ is a critical point of $\mathcal J$.
This finishes the proof of the theorem.
\end{proof}

\begin{remark}
In the theorem we neither assume nor obtain that the vector potential $\AA$ is divergence-free. However, this can easily be arranged by subtracting from $\AA$ the gradient of a function $\varphi$ satisfying the Poisson equation $\Delta \varphi=\div \AA$ with homogeneous Dirichlet condition on $\partial \Omega$.
One may also impose the condition that $\AA$ is divergence-free directly in the variational formulation. Indeed, after arriving at \eqref{interior variations} one finds that
$\curl(\curl \AA)-\alpha \curl \AA$ is the gradient of some function $\varphi \in C^{2,\alpha}(\Omega)$ with $\varphi=0$ on $\partial \Omega$. But then $\Delta \varphi =0$, so $\varphi$ vanishes. Alternatively, one can add a term $\int_{\Omega_{00}} |\div \AA|^2 \, \D V$ to the energy functional without changing the class of admissible vector fields.
\end{remark}

We have stated the variational principle in terms of diffeomorphisms $\v{F}$ from $\RR^3$ to itself and vector fields $\AA$ on $\RR^3$ in order to simplify  the presentation. Alternatively, we could  have stated it in terms of diffeomorphisms $\v{F}\colon \overline{D} \to \overline{\Omega}$ and
vector fields $\op{\AA}$ on $D$, with $\AA=T_{\v{F}} \op{\AA}$, in the notation of Lemma \ref{Teq}.
In other words, we replace (F1) by the condition
\begin{itemize}
	\item[(F1')] $\v{F}\colon \overline{D} \to \v{F}(\overline D)$ is a diffeomorphism with bounded partial derivatives and $\det{\textrm{D}\v{F}} > 0$ on $\overline{D}$.
\end{itemize}
Since $\op{\AA}$ is defined on the fixed domain $D$ there is then no need to extend the vector field outside of $\Omega$ (or $D$). The variational formulation can then be given in terms of the
functional
\[
\tilde{\mathcal J}(\op{\AA}, \v{F})=\mathcal J(T_{\v{F}} \op{\AA}, \v{F}(D))
\]
where $(\op{\AA}, \v{F})$ belongs to an affine subspace of $C^{2, \gamma}_{per}(D; \RR^3) \times C^{3, \gamma}_{loc}(D; \RR^3)$ defined by (F2), the periodicity condition on $\v{F}$ and the boundary conditions \eqref{condhat} and \eqref{condbothat}. Thus, critical points can simply be interpreted using the G\^ateaux (or Fr\'echet) derivative.
The formulas for the derivatives of the functionals are more complicated when expressed in terms of $\op{\AA}$ and $\delta \op{\AA}$. However, after a change of variables one may express them as in the proof of Theorem \ref{MainTh} without extending the vector fields outside the domain if
$\delta \AA$ is defined by
\[
\delta \AA=T_{\v{F}} \delta \op{\AA}+D_{\v{F}} T_{\v{F}} [\delta F] \op{\AA},
\]
where the last term is interpreted formally as if $\v{F}$ were a diffeomorphism on $\RR^3$ and $\op{\AA}$  defined on $\RR^3$.
This framework also gives a simple way of interpreting critical points of $\mathcal J$ as
critical points of $\mathcal E$ under the constraints that $\mathcal K$ and $\mathcal M$ are fixed ($\alpha$ and $\mu$ being the corresponding Lagrange multipliers) by a straightforward appeal to the implicit function theorem.
A third approach is to first consider variations of $\AA$ for a fixed $\Omega$ and then consider variations of $\Omega$ where $\AA=\AA_\Omega$ is a partial critical point; see  e.g.~\cite{BT, DZ}.
Finally, let us mention that in all of these approaches, one should really consider equivalence classes of parametrizations $\v{F}$. We ignore this point here to simplify the presentation, but it is in principle straightforward to take it into account; see e.g~\cite[Chapter 3]{DZ}.

\bigskip

\noindent {\bf Acknowledgments}. This project has received funding from the European Research Council (ERC) under the European Union's Horizon 2020 research and innovation programme (grant agreement no 678698).

\bibliographystyle{siam}
\bibliography{LokharuWahlen}

\begin{thebibliography}{10}

\bibitem{AR}
{\sc B.~F. Akers and J.~A. Reeger}, {\em Three-dimensional overturned traveling
  water waves}, Wave Motion, 68 (2017), pp.~210 -- 217.

\bibitem{A}
{\sc H.-D. Alber}, {\em Existence of three-dimensional, steady, inviscid,
  incompressible flows with nonvanishing vorticity}, Math. Ann., 292 (1992),
  pp.~493--528.

\bibitem{AK}
{\sc V.~I. Arnold and B.~A. Khesin}, {\em Topological methods in
  hydrodynamics}, vol.~125 of Applied Mathematical Sciences, Springer-Verlag,
  New York, 1998.

\bibitem{BA}
{\sc T.~Z. Boulmezaoud and T.~Amari}, {\em On the existence of non-linear
  force-free fields in three-dimensional domains}, Z. Angew. Math. Phys., 51
  (2000), pp.~942--967.

\bibitem{Brand}
{\sc L.~Brand}, {\em Vector and {T}ensor {A}nalysis}, John Wiley \& Sons, Inc.,
  New York; Chapman \& Hall Ltd., London, 1947.

\bibitem{BB}
{\sc B.~Buffoni and G.~R. Burton}, {\em On the stability of travelling waves
  with vorticity obtained by minimization}, NoDEA Nonlinear Differential
  Equations Appl., 20 (2013), pp.~1597--1629.

\bibitem{BGSW}
{\sc B.~Buffoni, M.~D. Groves, S.~M. Sun, and E.~Wahl{\'e}n}, {\em Existence
  and conditional energetic stability of three-dimensional fully localised
  solitary gravity-capillary water waves}, J. Differential Equations, 254
  (2013), pp.~1006--1096.

\bibitem{BGW}
{\sc B.~Buffoni, M.~D. Groves, and E.~Wahl{\'e}n}, {\em A variational reduction
  and the existence of a fully localised solitary wave for the
  three-dimensional water-wave problem with weak surface tension}, Arch.
  Rational Mech. Anal., 228 (2018), pp.~773--820.

\bibitem{BW}
{\sc B.~Buffoni and E.~Wahl{\'e}n}, {\em Steady three-dimensional rotational
  flows: an approach via two stream functions and {Nash}--{Moser} iteration}.
\newblock Preprint: {\tt https://arxiv.org/abs/1709.05957}, 2017.

\bibitem{BT}
{\sc G.~R. Burton and J.~F. Toland}, {\em Surface waves on steady perfect-fluid
  flows with vorticity}, Comm. Pure Appl. Math., 64 (2011), pp.~975--1007.

\bibitem{CW}
{\sc S.~Chandrasekhar and L.~Woltjer}, {\em On force-free magnetic fields},
  Proc. Nat. Acad. Sci. U.S.A., 44 (1958), pp.~285--289.

\bibitem{CSS}
{\sc A.~Constantin, D.~Sattinger, and W.~Strauss}, {\em Variational
  formulations for steady water waves with vorticity}, J. Fluid Mech., 548
  (2006), pp.~151--163.

\bibitem{CN}
{\sc W.~Craig and D.~P. Nicholls}, {\em Travelling two and three dimensional
  capillary gravity water waves}, SIAM J. Math. Anal., 32 (2000), pp.~323--359.

\bibitem{DZ}
{\sc M.~C. Delfour and J.-P. Zol\'esio}, {\em Shapes and geometries}, vol.~22
  of Advances in Design and Control, Society for Industrial and Applied
  Mathematics (SIAM), Philadelphia, PA, second~ed., 2011.
\newblock Metrics, analysis, differential calculus, and optimization.

\bibitem{EPS1}
{\sc A.~Enciso and D.~Peralta-Salas}, {\em Knots and links in steady solutions
  of the {E}uler equation}, Ann. of Math. (2), 175 (2012), pp.~345--367.

\bibitem{EPS2}
\leavevmode\vrule height 2pt depth -1.6pt width 23pt, {\em Existence of knotted
  vortex tubes in steady {E}uler flows}, Acta Math., 214 (2015), pp.~61--134.

\bibitem{EPS3}
\leavevmode\vrule height 2pt depth -1.6pt width 23pt, {\em Beltrami fields with
  a nonconstant proportionality factor are rare}, Arch. Ration. Mech. Anal.,
  220 (2016), pp.~243--260.

\bibitem{Freidberg}
{\sc J.~P. Freidberg}, {\em Ideal MHD}, Cambridge University Press, 2014.

\bibitem{GT}
{\sc D.~Gilbarg and N.~S. Trudinger}, {\em Elliptic partial differential
  equations of second order}, Classics in Mathematics, Springer-Verlag, Berlin,
  2001.
\newblock Reprint of the 1998 edition.

\bibitem{GH}
{\sc M.~D. Groves and M.~Haragus}, {\em A bifurcation theory for
  three-dimensional oblique travelling gravity-capillary water waves}, J.
  Nonlinear Sci., 13 (2003), pp.~397--447.

\bibitem{GM}
{\sc M.~D. Groves and A.~Mielke}, {\em A spatial dynamics approach to
  three-dimensional gravity-capillary steady water waves}, Proc. Roy. Soc.
  Edinburgh Sect. A, 131 (2001), pp.~83--136.

\bibitem{GS}
{\sc M.~D. Groves and S.-M. Sun}, {\em Fully localised solitary-wave solutions
  of the three-dimensional gravity-capillary water-wave problem}, Arch. Ration.
  Mech. Anal., 188 (2008), pp.~1--91.

\bibitem{GrovesWahlen07}
{\sc M.~D. Groves and E.~Wahl{\'e}n}, {\em Spatial dynamics methods for
  solitary gravity-capillary water waves with an arbitrary distribution of
  vorticity}, SIAM J. Math. Anal., 39 (2007), pp.~932--964.

\bibitem{GW1}
{\sc M.~D. Groves and E.~Wahl{\'e}n}, {\em Existence and conditional energetic
  stability of solitary gravity-capillary water waves with constant vorticity},
  Proc. Roy. Soc. Edinburgh Sect. A, 145 (2015), pp.~791--883.

\bibitem{IP1}
{\sc G.~Iooss and P.~Plotnikov}, {\em Asymmetrical three-dimensional travelling
  gravity waves}, Arch. Ration. Mech. Anal., 200 (2011), pp.~789--880.

\bibitem{IP2}
{\sc G.~Iooss and P.~I. Plotnikov}, {\em Small divisor problem in the theory of
  three-dimensional water gravity waves}, Mem. Amer. Math. Soc., 200 (2009),
  pp.~viii+128.

\bibitem{KNW}
{\sc R.~Kaiser, M.~Neudert, and W.~von Wahl}, {\em On the existence of
  force-free magnetic fields with small nonconstant {$\alpha$} in exterior
  domains}, Comm. Math. Phys., 211 (2000), pp.~111--136.

\bibitem{AL}
{\sc P.~Laurence and M.~Avellaneda}, {\em On {W}oltjer's variational principle
  for force-free fields}, J. Math. Phys., 32 (1991), pp.~1240--1253.

\bibitem{L}
{\sc J.~C. Luke}, {\em A variational principle for a fluid with a free
  surface}, J. Fluid Mech., 27 (1967), pp.~395--397.

\bibitem{Priest}
{\sc E.~Priest}, {\em Magnetohydrodynamics of the Sun}, Cambridge University
  Press, 2014.

\bibitem{RS}
{\sc J.~Reeder and M.~Shinbrot}, {\em Three-dimensional, nonlinear wave
  interaction in water of constant depth}, Nonlinear Anal., 5 (1981),
  pp.~303--323.

\bibitem{TX}
{\sc C.~Tang and Z.~Xin}, {\em Existence of solutions for three dimensional
  stationary incompressible {E}uler equations with nonvanishing vorticity},
  Chin. Ann. Math. Ser. B, 30 (2009), pp.~803--830.

\bibitem{Wahlen07}
{\sc E.~Wahl{\'e}n}, {\em A {H}amiltonian formulation of water waves with
  constant vorticity}, Lett. Math. Phys., 79 (2007), pp.~303--315.

\bibitem{Wahlen14}
{\sc E.~Wahl{\'e}n}, {\em Non-existence of three-dimensional travelling water
  waves with constant non-zero vorticity}, J. Fluid Mech., 746 (2014).

\bibitem{W1}
{\sc L.~Woltjer}, {\em A theorem on force-free magnetic fields}, Proc. Nat.
  Acad. Sci. U.S.A., 44 (1958), pp.~489--491.

\end{thebibliography}

\end{document}